\documentclass[12pt]{amsart}

\usepackage{amssymb,amsthm,amsmath}
\usepackage[numbers,sort&compress]{natbib}
\usepackage{color}
\usepackage{graphicx}
\usepackage{amsthm}
\usepackage{tikz}
\usepackage{amssymb,amsthm,amsmath}
\usepackage[numbers,sort&compress]{natbib}
\usepackage{color}
\usepackage{graphicx}
\usepackage{tikz}
\usepackage{xparse}

\title[Fermi isospectrality of separable periodic   operators ]{ Fermi isospectrality of discrete periodic Schr\"odinger operators  with separable potentials on $\mathbb{Z}^2$}

\author{Wencai Liu}
\address[W. Liu]{ Department of Mathematics, Texas A\&M University, College Station, TX 77843-3368, USA} \email{liuwencai1226@gmail.com; wencail@tamu.edu}



\theoremstyle{plain}
\newtheorem{theorem}{Theorem}[section]
\newcommand{\R}{\mathbb{R}}

\newtheorem{lemma}[theorem]{Lemma}

\newtheorem{remark}{Remark}
\newcommand{\C}{\mathbb{C}}

\newcommand{\Z}{\mathbb{Z}}

\theoremstyle{plain}
\newtheorem{definition}{Definition}
\newtheorem{conjecture}{Conjecture}

\begin{document}
	
	
	\begin{abstract}
		Let $\Gamma=q_1\mathbb{Z}\oplus q_2 \mathbb{Z} $ with $q_1\in \mathbb{Z}_+$ and $q_2\in\mathbb{Z}_+$.
		Let $\Delta+X$ be the discrete periodic  Schr\"odinger operator on $\mathbb{Z}^2$,
		where $\Delta$ is the discrete Laplacian and $X:\mathbb{Z}^2\to \C$ is $\Gamma$-periodic. In this paper, we develop tools from complex analysis to study the isospectrality of discrete periodic  Schr\"odinger operators. We prove that  if  two $\Gamma$-periodic potentials  $X$ and $Y$ are Fermi isospectral  and both  $X=X_1\oplus X_2$ and $Y=  Y_1\oplus Y_2$ are separable   functions,  then, up to a constant, one dimensional potentials  $X_j$  and $Y_j$ are Floquet isospectral, $j=1,2$. This allows us to prove that for any  non-constant  separable real-valued $\Gamma$-periodic potential,  the  Fermi variety $F_{\lambda}(V)/\mathbb{Z}^2$  is irreducible for any $\lambda\in \mathbb{C}$, which partially confirms a  conjecture  of Gieseker, Kn\"{o}rrer  and Trubowitz in the early 1990s.

	\end{abstract}
	
	\maketitle 
	\section{Introduction and main results}

	Given $q_j\in \Z_+$, $j=1,2,\cdots,d$,
	let $\Gamma=q_1\Z\oplus q_2 \Z\oplus\cdots\oplus q_d\Z $.
	We say that a function $V: \Z^d\to \C$ is  $\Gamma$-periodic (or just periodic)  if 
	for any $\gamma\in \Gamma$ and $n\in\Z^d$,  $V(n+\gamma)=V(n)$. 
	For $n=(n_1,n_2,\cdots,n_d)\in\Z^d$, denote by 	$||n||_1=\sum_{j=1}^d |n_j|$.
	Let  $\Delta$ be the discrete Laplacian on lattices  $\Z^d$, namely
	\begin{equation*}
	(\Delta u)(n)=\sum_{n'\in\Z^d, ||n^\prime-n||_1=1}u(n^\prime).
	\end{equation*}

		In the following, we  always assume that $q_j$, $j=1,2,\cdots, d$, are pairwise coprime  and $V$ is $\Gamma$-periodic.

	In this article we are interested in the isospectrality problem and irreducibility of Fermi varieties  of discrete periodic Schr\"odinger operators $\Delta+V$.  We refer readers to two survey articles ~\cite{ksurvey,lsur} for background and recent developments about the two topics.

	Let $\{\textbf{e}_j\}$, $j=1,2,\cdots d$,  be the   standard basis in $\Z^d$:
	\begin{equation*}
	\textbf{e}_1 =(1,0,\cdots,0),\textbf{e}_2 =(0,1,0,\cdots,0),\cdots, \textbf{e}_{d}=(0,0,\cdots,0,1).
	\end{equation*}
	
	\begin{definition}
		The   { \it Bloch variety} $B(V)$ of  $\Delta+V$ consists of all pairs $(k,\lambda )\in \C^{d+1}$ for which
		there exists a non-zero solution of the equation 
		\begin{equation} 
		(\Delta u)(n)+V(n) u(n)=\lambda u(n) ,n\in\Z^d,  \label{spect_0}
		\end{equation}
		satisfying the so called Floquet-Bloch  boundary condition
		\begin{equation} 
		u(n+q_j\textbf{e}_j)=e^{2\pi i k_j}u(n),j=1,2,\cdots,d, \text{ and }  n\in \Z^d, \label{Fl}
		\end{equation}
		where $k=(k_1,k_2,\cdots,k_d)\in \C^d$.

		Given  $\lambda\in \C$,
		the  Fermi surface (variety) $F_{\lambda}(V)$ is defined as the level set of the  Bloch variety:
		\begin{equation*}
		F_{\lambda}(V)=\{k: (k,\lambda)\in B(V)\}.
		\end{equation*}
	\end{definition}

We call $k=(k_1,k_2,\cdots,k_d)$ that appears in \eqref{Fl} quasi-momentum.
	One can see that both Fermi and Bloch varieties are analytic sets, in fact  algebraic sets after changing variables~\cite{ksurvey,lsur,liu1}.

		Our first interest is the isospectrality problems.
		
	Let $D_{V} (k)$ be the periodic operator $\Delta+V$ with the Floquet-Bloch boundary condition \eqref{Fl} (see Section \ref{S2} for the precise description of $D_{V} (k)$).
	Two $\Gamma$-periodic potentials $X$ and $Y$ are called 
	Floquet isospectral if  
	\begin{equation}\label{gfi}
	\sigma(D_{X} (k))= \sigma(D_{Y} (k)), \text{ for any } k \in\R^d.
	\end{equation}
	Two $\Gamma$-periodic potentials $X$ and $Y$ are called 
isospectral if  
	\begin{equation}\label{gfinew}
	\sigma(D_{X} (k))= \sigma(D_{Y} (k)) \text{  with  } k =0.
	\end{equation}
	
	Understanding when two periodic potentials $X$ and $Y$ are Floquet isospectral or isospectral is a fascinating subject and has been extensively studied
	~\cite{ksurvey,ERT84,MT76,kapiii,Kapi,Kapii,ERTII,gki,wa,gkii,gui90,eskin89}.

	In \cite{liu2021fermi}, the author introduced a new type of isospectrality: Fermi isospectrality. 
	\begin{definition}\label{fermiiso}	\cite{liu2021fermi}
		Let $X$ and $Y$ be two $\Gamma$-periodic functions. We say $X$ and $Y$ are  Fermi isospectral if 
		${F}_{\lambda_0} (X)={F}_{\lambda_0} (Y)$ for some $\lambda_0\in\C$. 
	\end{definition}
	It is not difficult to see that  two periodic  functions $X$ and $Y$ are Floquet isospectral if and only if
	Bloch varieties of $X$ and $Y$ are the same (or Fermi varieties of $X$ and $Y$ are the same for every $\lambda\in\C$)~\cite{liu2021fermi}.  So Fermi isospectrality is a ``hyperplane" version of Floquet isospectrality.

	In \cite{liu2021fermi}, the author proved several  rigidity theorems  of discrete periodic Schr\"odinger operators about separable functions.
	We say that a function $V$ on $\Z^d$ is $(d_1,d_2,\cdots,d_r)$ separable (or simply separable, denote it by $V=\bigoplus_{j=1}^rV_j$), where  $\sum_{j=1}^r d_j= d$ with $r\geq 2$, if there exist functions 
	$V_j$  on $\Z^{d_j}$, $j=1,2,\cdots,r$,   
	such that  for any $(n_1,n_2,\cdots,n_d)\in\Z^d$,
	\begin{align}
	V(n_1,n_2,\cdots,n_d)=&V_1(n_1,\cdots, n_{d_1})+V_2(n_{d_1+1},n_{d_1+2},\cdots,n_{d_1+d_2})\nonumber\\
	&+\cdots+V_r(n_{d_1+d_2+\cdots +d_{r-1}+1},\cdots,n_{d_1+d_2+\cdots +d_r}).\label{g61}
	\end{align}
	One of  rigidity theorems in \cite{liu2021fermi} states
	\begin{theorem}\label{thm}
		\cite{liu2021fermi}
		Let $d\geq 3$. 
		Assume   that  two separable  $\Gamma$-periodic potentials  $X=\bigoplus_{j=1}^rX_j$   and    $Y=\bigoplus_{j=1}^rY_j$ are Fermi isospectral.  
		Then, up to a constant,  lower dimensional decompositions  $V_j$  and $Y_j$ are Floquet isospectral, $j=1,2,\cdots,r$.
	\end{theorem}

	In the present work, we prove that the statement in Theorem \ref{thm} holds for dimension $d=2$. Namely,

	\begin{theorem}\label{thm1}
		Let $d=2$. 
		Assume that two  $\Gamma$-periodic potentials  $X$ and $Y$ are Fermi isospectral and both  $X=X_1\oplus X_2$ and $Y=  Y_1\oplus Y_2$ are separable. Then, up to a constant, one dimensional functions $X_j$  and $Y_j$ are Floquet isospectral, $j=1,2$.
	\end{theorem}
	
	\begin{remark}
		In Theorems \ref{thm} and \ref{thm1}, potentials are allowed to be complex-valued.
	\end{remark}
	Our second interest of this paper is the irreducibility of  Fermi varieties. Irreducibility of  Fermi varieties (also Bloch varieties) and  related applications such as embedded eigenvalues and spectral band edges have seen continuous progress in the  past 30 years~\cite{GKTBook,ktcmh90,bktcm91,bat1,batcmh92,ls,shi2,flscmp,kv06cmp,kvcpde20,shi1,flm,faust,dksjmp20}.

	Recently, the author introduced 	an algebraic method  and provided more general
	proofs of irreducibility  of Fermi varieties~\cite{liu1}.
	
	Denote by $[V]$ the average of $V$ over one periodicity cell, namely
	\begin{equation*}
	[V]=\frac{1}{q_1q_2\cdots q_d}\sum_{1\leq n_j\leq q_j\atop{{1\leq j\leq d}}}V(n_1,n_2,\cdots,n_d).
	\end{equation*}
	
	\begin{theorem}\label{gcf1}
		\cite{liu1}
		For any $d\geq 3$, the  Fermi variety  $F_{\lambda}(V)/\Z^d$ is irreducible for any $\lambda\in \C$.
		For $d=2$, the  Fermi variety $F_{\lambda}(V)/\Z^2$  is irreducible for any $\lambda\in \C$ except maybe  for $\lambda=[V] $ and  $F_{[V]}(V)/\Z^2$ has at most two irreducible components.
		
	\end{theorem}
	  
	Before \cite{liu1}, the irreducibility of  Fermi varieties at all  energy levels  for  $d=3$ and at all energy levels but finitely many $\lambda$ for $d=2$ was proved in \cite{GKTBook,batcmh92} by an different approach (compactification).

	Let $d=2$. 
	When the potential $V$ is a constant function, direct computation (e.g., see \cite{lsur}) implies that $F_{[V]}(V)/\Z^2$  has exactly two irreducible components.
	When the complex-valued functions are allowed, there exist non-constant complex valued  functions $V$ such that the Fermi variety is reducible at the energy level $[V]$ (e.g. \cite{fls22}).  

	However, for real-valued potentials, people believe  the constant potential is the only case that 
	the Fermi variety $F_{\lambda}(V)/\Z^2$  is   reducible  at  some energy level, which has been formulated as a conjecture  by Gieseker, Kn\"{o}rrer  and Trubowitz in the early 1990s~\cite{GKTBook}.
	
	{\bf Conjecture 1:} ~\cite[p.43]{GKTBook}
	Assume  that $V$  is  a non-constant real-valued periodic potential on $\Z^2$. Then the  Fermi variety $F_{\lambda}(V)/\Z^2$  is irreducible for any $\lambda\in \C$.

	Theorem \ref{thm1} allows us to confirm the Conjecture 1  for separable potentials.
	\begin{theorem}\label{thm2}	
		Assume  that $V$  is  a non-constant separable real-valued periodic potential on $\Z^2$. Then the  Fermi variety $F_{\lambda}(V)/\Z^2$  is irreducible for any $\lambda\in \C$.
	\end{theorem}

	The irreducibility of Fermi variety and  Fermi isospectrality  of discrete periodic Schr\"odinger operators (dimension $d\geq3$) are well understood in two recent papers~\cite{liu1,liu2021fermi}. Besides Theorem \ref{thm}, there are  other Fermi isospectrality results in~\cite{liu2021fermi} for dimension $d\geq 3$. However, approaches in \cite{liu2021fermi} can not be extended to dimension $d=2$ since there are not enough free variables available.  
	For irreducibility results of the Fermi variety  in Theorem \ref{gcf1}, the proof for  $d=2$ is   more difficult than that for $d\geq 3$.
	For continuous periodic Schr\"odinger operators,  
	B{\"a}ttig,  Kn\"orrer and Trubowitz \cite{bktcm91} proved  the  irreducibility of Fermi varieties and a rigidity theorem  of separable functions in   dimension three.
	However, the proof in  \cite{bktcm91} 
	does not work for dimension $d=2 $.
	For discrete periodic Schr\"odinger operators on $\Z^d$ with $d\geq 3$,  the Fermi variety $F_{\lambda}(V)/\Z^d$ for any complex-valued potential is irreducible at any energy level $\lambda$ (see Theorem \ref{gcf1}). 
	For $d=2$, there are many complex-valued  potentials $V$ such that the Fermi variety $F_{\lambda}(V)/\Z^2$  has two irreducible components  at the average energy level $[V]$~\cite{fls22}.
	
 Finally, we want to comment that dimension two is the transition of Fermi isospectrality problems of periodic Shcr\"odinger operators.  For $d=1$, it does not make sense to study Fermi isospectrality  since for any periodic potential $V$, $F_{\lambda_0}(V)$ contains at most two points. For $d=2$ and any periodic potential $Y$, all periodic potentials $X$ such that  $X$ and $Y$ are Fermi isospectral  at $\lambda_0$ (namely $F_{\lambda_0}(X)=F_{\lambda_0}(Y)$)  is an algebraic set with  at least one dimension \cite{fls22}.  For $d\geq 3$ and any periodic potential $Y$, 
 all periodic potentials $X$ such that  $X$ and $Y$ are Fermi isospectral  at $\lambda_0$ could be an algebraic set with  zero dimension~\cite{fls22}.   
	
	All  evidence  above seems to indicate that  when $d=2$,  problems related to Fermi varieties are  special (often more challenging).    

	In  this paper,
	we present a novel  approach to study the  Fermi isospectrality  of discrete periodic Schr\"odinger operators. As in \cite{liu2021fermi}, we focus on the study of a family of  Laurent polynomials  whose zero sets  are Fermi varieties after changing variables.
	Our strategy is to  develop  tools from complex analysis to study the eigenvalue problems of \eqref{spect_0} and \eqref{Fl} (or  \eqref{gei1} and \eqref{gei2}) with complexified quasi-momenta. One needs to relabel spectral band functions of one dimensional periodic Schr\"odinger operators based on asymptotics of  eigenvalues and show that  those functions are analytic with respect to  quasi-momenta in an  appropriate domain. Applying Rouche's Theorem, one sees that for any two one-dimensional $q$-periodic  potentials with the same average,  there exist  $q$ choices of  quasi-momenta such that  for those quasi-momenta,   labelled   eigenvalues of two potentials only differ by a (same) constant. 
	This enables us to show that separable components   of Fermi  isospectrality potentials with respect to one coordinate are Floquet isospectral and hence remaining separable components  are Floquet isospectral as well.

	The rest of this paper is organized as follows. In Section \ref{S2}, we  recall some basics for Fermi varieties. 
	In Section \ref{Sone}, we study one dimensional periodic Schr\"odinger operators.
	Section \ref{Smain} is devoted to proving Theorems \ref{thm1} and \ref{thm2}.

	\section{ Basics of Fermi varieties }\label{S2}

	Let  $\C^{\star}=\C\backslash \{0\}$ and $z=(z_1,z_2,\cdots,z_d)$. 
	For any $z\in (\C^{\star})^d$, consider the  equation
	\begin{equation}\label{gei1}
	(\Delta+V) u=\lambda u  
	\end{equation} 
	with   the boundary condition
	\begin{equation}\label{gei2}
	u(n+q_j\textbf{e}_j)=z_j u(n), j=1,2,\cdots,d,  \text{ and } n\in \Z^d,
	\end{equation}

	Introduce a fundamental domain $W$ for $\Gamma$:
	\begin{equation*}
	W=\{n=(n_1,n_2,\cdots,n_d)\in\Z^d: 0\leq n_j\leq q_{j}-1, j=1,2,\cdots, d\}.
	\end{equation*}
	By writing out $\Delta +V$  as acting on the  $Q=q_1q_2\cdots q_d$ dimensional space $\{u(n),n\in W\}$, 
	the  equation \eqref{gei1} with boundary condition \eqref{gei2} (\eqref{spect_0} and \eqref{Fl})
	translates into the eigenvalue problem for a  $Q\times Q$ matrix $\mathcal{D}_V(z)$ ($D_V(k)$).

	Let 
	\begin{equation}\label{g16}
	\mathcal{P}_V(z,\lambda)=\det(\mathcal{D}_V(z)-\lambda I), {P}_V(k,\lambda)=\det({D}_V(k)-\lambda I).
	\end{equation}
	We remark that $\mathcal{D}_V(z)$ and $D_V(k)$ ($\mathcal{P}_V(z,\lambda)$ and $\mathcal{P}_V(k,\lambda)$) are the same under the relations $z_j=e^{2\pi i k_j}$, $j=1,2,\cdots, d$.
	
	{\bf Example 1:}
	When $d=1$, 
	the equation $(\Delta+V)u=\lambda u$ with the Floquet-Bloch boundary condition $u(n+q)=z u(n)$, $z\in\C^\star$, can be reduced to   an eigenvalue problem of  a $q\times q$ matrix:
	\begin{equation}\label{g6}
	\mathcal{D}_V (z )= \begin{pmatrix}
	V(1)  &  1  & 0&   \cdots & 0 & z^{-1} \\
	1 & V(2)  & 1 & \cdots   & 0& 0 \\
	0& 1 & V(3)&   \cdots& 0 & 0 \\
	\vdots & \vdots & \vdots & \ddots & \vdots & \vdots  \\
	0 & 0 & 0 & \cdots   & V(q-1)  & 1 \\
	z & 0 & 0& \cdots  & 1 & V(q) \\
	\end{pmatrix}.
	\end{equation}
	
	We have the following {\bf Basic Facts}:
	\begin{enumerate}
		\item Fermi variety is given by 
		\begin{equation}\label{g11}
		F_{\lambda}(V) =\{k\in\C^d: {P}_V(k,\lambda) =0\}.
		\end{equation}	 
		\item Two periodic  functions $X$ and $Y$ are Floquet isospectral if and only if
		\begin{equation}
		\mathcal{P}_X(z,\lambda)=	\mathcal{P}_Y(z,\lambda).
		\end{equation}
	\end{enumerate}

	\section{One dimensional discrete periodic Schr\"odinger operators }\label{Sone}

	In this section, we  study one dimensional  discrete periodic Schr\"odinger operators $\Delta+V$:
	\begin{equation*}
	((\Delta +V)u)(n) =u(n+1)+u(n-1)+V(n) u(n), n\in\Z,
	\end{equation*}
	where $V$ is a  periodic function on $\Z$, namely, $V(n+q)=V(n),n\in\Z$ for some positive integer $q$.  
	
	In the following, we say $z$ is large if $|z|$ is large.
	
	By \eqref{g6} in Example 1,  $\mathcal{P}_V(z,\lambda)-(-1)^{q+1}z-(-1)^{q+1}z^{-1}$ is independent of variable $z$. So
	let $	\hat{\mathcal{P}}_V(\lambda)$ be such that
	\begin{equation}\label{g4}
	\mathcal{P}_V(z,\lambda)=	\hat{\mathcal{P}}_V(\lambda)+(-1)^{q+1}z+(-1)^{q+1}z^{-1}.
	\end{equation}
	By \eqref{g6}, one has that (recall that $[V]=\frac{1}{q}(\sum_{j=1}^q V(j))$)
	\begin{equation}\label{g7}
	{\hat{\mathcal{P}}}_{V}(z, \lambda) =(-1)^q\lambda^q-(-1)^qq[V]\lambda^{q-1}  +\text{ lower order terms of } \lambda.
	\end{equation}
	Fixing $z\in\C$, solve the algebraic equation 
	\begin{equation}\label{g15}
	{\mathcal{P}}_{V}(z^q, \lambda) =0.
	\end{equation}
	By \eqref{g4} and \eqref{g7},
	there exist solutions $\lambda_V^l(z)$ of equation \eqref{g15}, $l=0,1,2,\cdots, q-1$ such that  $\lambda^l(z)$ is analytic in $\Omega=\{z\in \C: |z|>R\}$ with large $R$ (the largeness only depends on the  potential $V$). Moreover,  $\lambda_V^l(z)$, $l=0,1,2,\cdots, q-1$ have the following representations  in  Laurent series,
	\begin{equation}\label{g8}
	\lambda_{V}^l(z)= e^{2\pi \frac{l}{q} i} z+ [V]+\sum_{k=1}^{\infty}\frac{a_k(V)}{z^k},
	\end{equation}
	where the coefficient $a_k(V)$ depends on $V$.
	
	Clearly, for any large  $z$, $\lambda_{V}^l(z)$, $l=0,1,2,\cdots, q-1$, are  eigenvalues of $\mathcal{D}(z^q)$, and hence 
	\begin{equation}\label{g2}
	{\mathcal{P}}_V(z^q,\lambda)=\det(\mathcal{D}(z^q)-\lambda I)=\prod_{l=0}^{q-1}(\lambda_V^l(z)-\lambda).
	\end{equation}
	\begin{lemma}\label{lem2}
		Assume $[V]=[\tilde{V}]$,  and $V$ and $\tilde{V}$ are not Floquet isospectral. Then there exist  sufficiently large $R>0$ and  small $\epsilon>0$ such that for any $ \eta\in\C$ with $0<|\eta|<\epsilon$ and $ l=0,1,2,\cdots, q-1$, the  equation  $\lambda_{V}^l(z)=\lambda_{\tilde{V}}^l(z)+\eta$ has  at least one solution in $ \{z\in\C: |z|\geq R\}$.
	\end{lemma}
	\begin{proof}
		Fix any $l\in\{0,1,\cdots,q-1\}$.
		By \eqref{g8}, one has that for any large $z$,
		\begin{equation}\label{g3}
		\lambda_{V}^l(z)= e^{2\pi \frac{l}{q} i} z+ [V]+\sum_{k=1}^{\infty}\frac{a_k(V)}{z^k},
		\end{equation}
		and
		\begin{equation}\label{g3new}
		\lambda_{\tilde{V}}^l(z)= e^{2\pi \frac{l}{q} i} z+ [\tilde{V}]+\sum_{k=1}^{\infty}\frac{a_k(\tilde{V})}{z^k},
		\end{equation}
		If $	\lambda_{V}^l(z)=	\lambda_{\tilde{V}}^l(z)$ for any large  $z$, then $V$ and $\tilde{V}$ must be Flqouet isospectral, which contradicts the assumption. So functions 
		$	\lambda_{V}^l(z)$ and $	\lambda_{\tilde{V}}^l(z)$ are not identical.

		When $a_1(V)\neq a_1(\tilde{V})$, let $k_0=1$. Otherwise, 
		let $k_0\in\Z_+$ be such that $a_k(V)=a_k(\tilde{V})$ for any $k<k_0$ and  $a_{k_0}(V)\neq a_{k_0}(\tilde{V})$.
		Consider a ball $B_{\epsilon_1}=\{z\in\C: |z|\leq \epsilon_1\}$ with a small $\epsilon_1>0$. 
		Define an analytic function $f(z)$ in  $ B_{\epsilon_1}$ by 
		\begin{equation}\label{g27}
		f(z) =\sum_{k=k_0}^{\infty}( a_k(V)z^{k}-a_k(\tilde{V})z^{k}).
		\end{equation}
		Then for any large $z$, 
		\begin{equation}\label{g10}
		f(z^{-1})=\lambda_{V}^l(z)-	\lambda_{\tilde V}^l(z).
		\end{equation}
		
		Let 
		\begin{equation*}
		f_1(z)= a_k(V)z^{k_0}-a_k(\tilde{V})z^{k_0},
		\end{equation*}
		and
		\begin{equation*}
		f_2(z)=\sum_{k=k_0+1}^{\infty} (a_k(V)z^{k}-a_k(\tilde{V})z^{k})
		\end{equation*}

		For any $z\in \partial B_{\epsilon_1}=\{z\in\C: |z|= \epsilon_1\}$, one has that
		\begin{equation}\label{g29}
		|	f_1(z)| \geq   |f_2(z)| +\frac{|a_{k_0}(V)-a_{k_0}(\tilde{V})|}{2}\epsilon_1^{k_0}.
		\end{equation}
		Let $\epsilon$ be sufficiently small (depending on  $\epsilon_1$). Choose any $\eta$ with $0<|\eta|\leq \epsilon$.
		By \eqref{g29}, one has that 
		for any $z\in \partial B_{\epsilon_1}$,  
		\begin{equation}\label{g30}
		|	f_1(z)-\eta| >  |f_2(z)|.
		\end{equation}
		By Rouche's theorem,  $f_1(z)-\eta=0$ and $f(z)-\eta=f_1(z)+f_2(z)-\eta=0$  have the same number of zeros (counting multiplicity) in $\{z\in\C: |z|< \epsilon_1\}$. This particularly implies that $f(z)=\eta$ has at least one non-zero solution in $\{z\in\C: |z|< \epsilon_1\}$.
		Now Lemma \ref{lem2} follows from \eqref{g10}.
	\end{proof}
	\section{Proof  of Theorems \ref{thm1} and \ref{thm2}}\label{Smain}

	\begin{lemma}\label{lem1} \cite[Lemma 2.3]{liu2021fermi}
		Assume $F_{\lambda_0}(X)=F_{\lambda_0}(Y)$. Then for any $z\in(\C^{\star})^2$, 
		\begin{equation}
		\mathcal{P}_X(z,\lambda_0)=\mathcal{P}_Y(z,\lambda_0).
		\end{equation}
	\end{lemma}
	\begin{lemma}\label{lem3}
		Assume that separable functions
		$X=X_1\oplus X_2$ and $Y=Y_1\oplus Y_2$ are Fermi isospectral. Assume that  $X_2$ and $Y_2$ are Floquet isospectral. Then 
		$X_1$ and $Y_1$ are Floquet isospectral.
	\end{lemma}
	\begin{proof}
		Recall that $z=(z_1,z_2)$.
		By the assumption  that $F_{\lambda_0} (X)=F_{\lambda_0} ({Y})$ and Lemma \ref{lem1}, one has that 
		\begin{equation*} 
		\mathcal{P}_X(z,\lambda_0)=	\mathcal{P}_Y(z,\lambda_0),
		\end{equation*}
		and hence
		\begin{equation}\label{g19}
		\mathcal{P}_X(z_1,z_2^{q_2},\lambda_0)=	\mathcal{P}_Y(z_1,z_2^{q_2},\lambda_0).
		\end{equation}
		
		Using the fact that both $X$ and $Y$ are  separable, one has that   for any large $z_2$,
		\begin{align}
		\mathcal{P}_X(z_1,z_2^{q_2},\lambda_0)&= \prod_{l=0}^{q_2-1}  \mathcal{P}_{X_1} (z_1,  -\lambda_{X_2}^l(z_2)+\lambda_0)\nonumber\\
		&=	\prod_{l=0}^{q_2-1} (\hat{\mathcal{P}}_{X_1}( -\lambda_{X_2}^l(z_2)+\lambda_0)+(-1)^{q_1+1}z_1+(-1)^{q_1+1}z_1^{-1}),\label{g17}
		\end{align}
		
		and 
		\begin{align}
		\mathcal{P}_Y(z_1,z_2^{q_2},\lambda_0)&= \prod_{l=0}^{q_2-1}  \mathcal{P}_{Y_1} (z_1, - \lambda_{Y_2}^l(z_2)+\lambda_0)	\nonumber\\
		&=	\prod_{l=0}^{q_2-1} (\hat{\mathcal{P}}_{Y_1}(- \lambda_{Y_2}^l(z_2)+\lambda_0)+(-1)^{q_1+1}z_1+(-1)^{q_1+1}z_1^{-1}).\label{g18}
		\end{align}
		By \eqref{g8}, \eqref{g19}, \eqref{g17}, \eqref{g18}, and the unique factorization theorem (using $(-1)^{q_1+1}z_1+(-1)^{q_1+1}z_1^{-1}$ as a variable), one has that  for any $l=0,1,\cdots, q_2-1$, 
		\begin{equation}\label{gj84}
		\hat{\mathcal{P}}_{X_1}( -\lambda_{X_2}^l(z_2)+\lambda_0)= \hat{\mathcal{P}}_{Y_1}( -\lambda_{Y_2}^l(z_2)+\lambda_0), 
		\end{equation}
		and 
		\begin{equation}\label{gj841}
		\mathcal{P}_{X_1} (z_1,  -\lambda_{X_2}^l(z_2)+\lambda_0)= \mathcal{P}_{Y_1} (z_1, - \lambda_{Y_2}^l(z_2)+\lambda_0).
		\end{equation}
		
		Since 	$ X_2$ and $Y_2$ are Floquet isospectral, we know that  for any large $z_2$,
		\begin{equation}\label{g12}
		\lambda_{X_2}^l(z_2)=\lambda_{Y_2}^l(z_2), l=0, 1,2,\cdots, q_2-1.
		\end{equation}
		By \eqref{gj841} and   \eqref{g12}, one has that  for any $z_1\in \C^\star$ and $\lambda\in\C$,
		\begin{equation}\label{g20}
		\mathcal{P}_{X_1} (z_1, \lambda)= \mathcal{P}_{Y_1} (z_1, \lambda).
		\end{equation}
		By \eqref{g20}	and basic fact (2) appearing at the end of Section \ref{S2}, we conclude that $X_1$ and $Y_1$ are Floquet isospectral.

	\end{proof}

	\begin{proof}[\bf Proof of Theorem \ref{thm1} ]
		
		Without loss of generality, assume that $q_2>q_1$ and $ [X_2]=[Y_2]=0$.
		 If $X_2$ and $Y_2$ are Floquet isospectral, then Theorem \ref{thm1} follows from Lemma \ref{lem3}. So  we assume  $X_2$ and $Y_2$   are not  Floquet isospectral.
		
		Applying  Lemma \ref{lem2} with $V=X_2$, $\tilde{V}=Y_2$ and $q=q_2$,  there exist  $\eta$ and large $x_l$, $l=0,1,2,\cdots q_2-1$,   such that 
		
		\begin{equation}\label{g13}
		\lambda_{X_2}^l(x_l)=	\lambda_{Y_2}^l(x_l)+\eta.
		\end{equation}
		By \eqref{gj84} and \eqref{g13}, we have that  for any $ l=0,1,2,\cdots, q_2-1$,
		\begin{equation}\label{g14}
		\hat{\mathcal{	P}}_{X_1} (	-\lambda_{Y_2}^l(x_l)-\eta+\lambda_0) =\hat{	\mathcal{P}}_{Y_1} (	-\lambda_{Y_2}^l(x_l)+\lambda_0).
		\end{equation}
		Since both $ \hat{\mathcal{	P}}_{X_1} (	\lambda-\eta) $ and $\hat{\mathcal{P}}_{Y_1} (	\lambda)$ are polynomials of $\lambda$ with degree $q_1$, by \eqref{g14} and the fact that $q_2>q_1$, one has that 
		\begin{equation}\label{g26}
		\hat{\mathcal{	P}}_{X_1} (	\lambda-\eta) =	\hat{\mathcal{	P}}_{X_1} (	\lambda),\lambda\in \C.
		\end{equation}
		This implies that 
		$X_1$ and $Y_1$ are Floquet isospectral up to a constant (by letting $\eta\to 0$ in \eqref{g26}, we can indeed show that 	$X_1$ and $Y_1$ are Floquet isospectral. This is because we  have already  shifted  the constant by setting $ [X_2]=[Y_2]=0$).  Now Theorem \ref{thm1} follows from Lemma \ref{lem3} (exchange $X_1$ and $X_2$, and $Y_1$ and $Y_2$). 
	\end{proof}

	Denote by ${\bf 0} $ the zero function on $\Z^2$.
	
	\begin{proof}[\bf Proof of Theorem \ref{thm2}]
		Assume that  $F_{\lambda}(V)$ is reducible at some $\lambda=\lambda_0$.
		By  Remark 4 in \cite{liu1} (also Theorem 2.5 in \cite{liu2021fermi}),   $\lambda_0=[V]$ and 	$\mathcal{P}_V(z,\lambda_0)=\mathcal{P}_{{\bf 0}}(z,0)$. Therefore	$\mathcal{P}_{V-[V]}(z,0)=\mathcal{P}_{{\bf 0}}(z,0)$.
		By Theorem \ref{thm1}, we have that $V$  and a constant potential  are Floquet isospectral.  Therefore,  Ambarzumian-type theorem (e.g. \cite{heflmp}) implies   $V$ is constant. This contradicts the assumption.
	\end{proof}

	\section*{Acknowledgments}
	
	This research was supported by NSF DMS-2000345 and  DMS-2052572.

	\bibliographystyle{abbrv} 
	\bibliography{absence}

\begin{thebibliography}{10}

\bibitem{bat1}
D.~B\"{a}ttig.
\newblock A directional compactification of the complex {F}ermi surface and
  isospectrality.
\newblock In {\em S\'{e}minaire sur les \'{E}quations aux {D}\'{e}riv\'{e}es
  {P}artielles, 1989--1990}, pages Exp. No. IV, 11. \'{E}cole Polytech.,
  Palaiseau, 1990.

\bibitem{batcmh92}
D.~B\"{a}ttig.
\newblock A toroidal compactification of the {F}ermi surface for the discrete
  {S}chr\"{o}dinger operator.
\newblock {\em Comment. Math. Helv.}, 67(1):1--16, 1992.

\bibitem{bktcm91}
D.~B\"{a}ttig, H.~Kn\"{o}rrer, and E.~Trubowitz.
\newblock A directional compactification of the complex {F}ermi surface.
\newblock {\em Compositio Math.}, 79(2):205--229, 1991.

\bibitem{dksjmp20}
N.~Do, P.~Kuchment, and F.~Sottile.
\newblock Generic properties of dispersion relations for discrete periodic
  operators.
\newblock {\em J. Math. Phys.}, 61(10):103502, 19, 2020.

\bibitem{eskin89}
G.~Eskin.
\newblock Inverse spectral problem for the {S}chr\"{o}dinger equation with
  periodic vector potential.
\newblock {\em Comm. Math. Phys.}, 125(2):263--300, 1989.

\bibitem{ERT84}
G.~Eskin, J.~Ralston, and E.~Trubowitz.
\newblock On isospectral periodic potentials in {${\bf R}\sp{n}$}.
\newblock {\em Comm. Pure Appl. Math.}, 37(5):647--676, 1984.

\bibitem{ERTII}
G.~Eskin, J.~Ralston, and E.~Trubowitz.
\newblock On isospectral periodic potentials in {${\bf R}\sp{n}$}. {II}.
\newblock {\em Comm. Pure Appl. Math.}, 37(6):715--753, 1984.

\bibitem{fls22}
M.~Faust, W.~Liu, and F.~Sottile.
\newblock In preparation.
\newblock 2022.

\bibitem{faust}
M.~Faust and F.~Sottile.
\newblock Critical points of discrete periodic operators.
\newblock {\em arXiv preprint arXiv:2206.13649}, 2022.

\bibitem{flm}
J.~Fillman, W.~Liu, and R.~Matos.
\newblock Irreducibility of the bloch variety for finite-range
  {S}chr{\"o}dinger operators.
\newblock {\em J. Funct. Anal. to appear.}

\bibitem{flscmp}
L.~Fisher, W.~Li, and S.~P. Shipman.
\newblock Reducible {F}ermi surface for multi-layer quantum graphs including
  stacked graphene.
\newblock {\em Comm. Math. Phys.}, 385(3):1499--1534, 2021.

\bibitem{GKTBook}
D.~Gieseker, H.~Kn\"{o}rrer, and E.~Trubowitz.
\newblock {\em The geometry of algebraic {F}ermi curves}, volume~14 of {\em
  Perspectives in Mathematics}.
\newblock Academic Press, Inc., Boston, MA, 1993.

\bibitem{gki}
C.~S. Gordon and T.~Kappeler.
\newblock On isospectral potentials on tori.
\newblock {\em Duke Math. J.}, 63(1):217--233, 1991.

\bibitem{gkii}
C.~S. Gordon and T.~Kappeler.
\newblock On isospectral potentials on flat tori. {II}.
\newblock {\em Comm. Partial Differential Equations}, 20(3-4):709--728, 1995.

\bibitem{gui90}
V.~Guillemin.
\newblock Inverse spectral results on two-dimensional tori.
\newblock {\em J. Amer. Math. Soc.}, 3(2):375--387, 1990.

\bibitem{heflmp}
B.~Hatino\u{g}lu, J.~Eakins, W.~Frendreiss, L.~Lamb, S.~Manage, and A.~Puente.
\newblock Ambarzumian-type problems for discrete {S}chr\"{o}dinger operators.
\newblock {\em Complex Anal. Oper. Theory}, 15(8):Paper No. 118, 13, 2021.

\bibitem{Kapi}
T.~Kappeler.
\newblock On isospectral periodic potentials on a discrete lattice. {I}.
\newblock {\em Duke Math. J.}, 57(1):135--150, 1988.

\bibitem{Kapii}
T.~Kappeler.
\newblock On isospectral potentials on a discrete lattice. {II}.
\newblock {\em Adv. in Appl. Math.}, 9(4):428--438, 1988.

\bibitem{kapiii}
T.~Kappeler.
\newblock Isospectral potentials on a discrete lattice. {III}.
\newblock {\em Trans. Amer. Math. Soc.}, 314(2):815--824, 1989.

\bibitem{ktcmh90}
H.~Kn\"{o}rrer and E.~Trubowitz.
\newblock A directional compactification of the complex {B}loch variety.
\newblock {\em Comment. Math. Helv.}, 65(1):114--149, 1990.

\bibitem{ksurvey}
P.~Kuchment.
\newblock An overview of periodic elliptic operators.
\newblock {\em Bull. Amer. Math. Soc. (N.S.)}, 53(3):343--414, 2016.

\bibitem{kvcpde20}
P.~Kuchment and B.~Vainberg.
\newblock On absence of embedded eigenvalues for {S}chr\"{o}dinger operators
  with perturbed periodic potentials.
\newblock {\em Comm. Partial Differential Equations}, 25(9-10):1809--1826,
  2000.

\bibitem{kv06cmp}
P.~Kuchment and B.~Vainberg.
\newblock On the structure of eigenfunctions corresponding to embedded
  eigenvalues of locally perturbed periodic graph operators.
\newblock {\em Comm. Math. Phys.}, 268(3):673--686, 2006.

\bibitem{ls}
W.~Li and S.~P. Shipman.
\newblock Irreducibility of the {F}ermi surface for planar periodic graph
  operators.
\newblock {\em Lett. Math. Phys.}, 110(9):2543--2572, 2020.

\bibitem{liu2021fermi}
W.~Liu.
\newblock Fermi isospectrality for discrete periodic {S}chr{\"o}dinger
  operators.
\newblock {\em Comm. Pure Appl. Math. to appear}.

\bibitem{liu1}
W.~Liu.
\newblock Irreducibility of the {F}ermi variety for discrete periodic
  {S}chr\"{o}dinger operators and embedded eigenvalues.
\newblock {\em Geom. Funct. Anal.}, 32(1):1--30, 2022.

\bibitem{lsur}
W.~Liu.
\newblock Topics on {F}ermi varieties of discrete periodic {S}chr\"{o}dinger
  operators.
\newblock {\em J. Math. Phys.}, 63(2):Paper No. 023503, 13, 2022.

\bibitem{MT76}
H.~P. McKean and E.~Trubowitz.
\newblock Hill's operator and hyperelliptic function theory in the presence of
  infinitely many branch points.
\newblock {\em Comm. Pure Appl. Math.}, 29(2):143--226, 1976.

\bibitem{shi1}
S.~P. Shipman.
\newblock Eigenfunctions of unbounded support for embedded eigenvalues of
  locally perturbed periodic graph operators.
\newblock {\em Comm. Math. Phys.}, 332(2):605--626, 2014.

\bibitem{shi2}
S.~P. Shipman.
\newblock Reducible {F}ermi surfaces for non-symmetric bilayer quantum-graph
  operators.
\newblock {\em J. Spectr. Theory}, 10(1):33--72, 2020.

\bibitem{wa}
A.~Waters.
\newblock Isospectral periodic torii in dimension 2.
\newblock {\em Ann. Inst. H. Poincar\'{e} Anal. Non Lin\'{e}aire},
  32(6):1173--1188, 2015.

\end{thebibliography}
	
\end{document}